\documentclass[manuscript,11pt,natbib=false]{acmart}

\AtBeginDocument{%
  \providecommand\BibTeX{{%
    \normalfont B\kern-0.5em{\scshape i\kern-0.25em b}\kern-0.8em\TeX}}}

\setcopyright{none}

\usepackage{microtype}
\usepackage{graphicx}
\usepackage{amsfonts}
\usepackage{amsmath}
\usepackage{color}
\usepackage{mathrsfs}
\usepackage{tikz-cd}
\usepackage{centernot}
\usepackage{algorithm,algpseudocode}

\makeatletter
\DeclareRobustCommand{\cev}[1]{%
  \mathpalette\do@cev{#1}%
}
\newcommand{\do@cev}[2]{%
  \fix@cev{#1}{+}%
  \reflectbox{$\m@th#1\vec{\reflectbox{$\fix@cev{#1}{-}\m@th#1#2\fix@cev{#1}{+}$}}$}%
  \fix@cev{#1}{-}%
}
\newcommand{\fix@cev}[2]{%
  \ifx#1\displaystyle
    \mkern#23mu
  \else
    \ifx#1\textstyle
      \mkern#23mu
    \else
      \ifx#1\scriptstyle
        \mkern#22mu
      \else
        \mkern#22mu
      \fi
    \fi
  \fi
}

\makeatother

\begin{document}

\title{Tight Localizations of Feedback Sets}

\author{Michael Hecht}
\orcid{https://orcid.org/0000-0001-9214-8253}
\authornotemark[2]
\email{hecht@mpi-cbg.de}

\authornote{All authors contributed equally to this research.}

\author{Krzysztof Gonciarz}
\orcid{https://orcid.org/0000-0001-9054-8341}
\authornotemark[2]
\email{gonciarz@mpi-cbg.de}

\author{Szabolcs Horv\'at}
\orcid{https://orcid.org/0000-0002-3100-523X}
\email{horvat@mpi-cbg.de}
\authornotemark[2]
\authornotemark[3]

\affiliation{%
  \institution{ \\ $\dagger$ Max Planck Institute of Molecular Cell Biology and Genetics, Center for Systems~Biology~Dresden}
  \streetaddress{Pfotenhauerstrasse 108}
  \city{Dresden}
  \state{Germany}
  \postcode{01307}
}

\affiliation{%
  \institution{ \\ $\ddagger$ Max Planck Institute for the Physics of Complex Systems}
  \streetaddress{N\"{o}thnitzerstrasse 38}
  \city{Dresden}
  \state{Germany}
  \postcode{01187}
}

\renewcommand{\shortauthors}{Hecht, Gonciarz, Horv\'at}

\newcommand\eg{\textit{e.g.\ }}
\newcommand\etc{\textit{etc}}
\newcommand\p{\partial}
\newcommand\Oc{\mathcal{O}}

\newcommand{\R}{\mathbb{R}}
\newcommand{\A}{\mathbb{A}}
\newcommand{\M}{\mathbb{M}}
\newcommand{\T}{\mathbb{T}}
\newcommand{\LL}{\mathscr{L}}

\newcommand{\Ee}{\mathbb{E}}
\newcommand{\Ff}{\mathbb{F}}
\newcommand{\Hm}{\mathbb{H}}
\newcommand{\X}{\mathbb{X}}
\newcommand{\PP}{\mathbb{P}}
\newcommand{\mS}{\mathbb{S}}
\newcommand{\N}{\mathbb{N}}
\newcommand{\W}{\mathbb{W}}
\newcommand{\VV}{\mathbb{V}}
\newcommand{\KK}{\mathbb{K}}
\newcommand{\I}{\mathbb{I}}
\newcommand{\Q}{\mathbb{Q}}
\newcommand{\Z}{\mathbb{Z}}
\newcommand{\Rn}{\mathbb{R}^{2n}}
\newcommand{\Tn}{\mathbb{T}^{2n}}
\newcommand{\Zn}{\mathbb{Z}^{2n}}
\newcommand{\ee}{\varepsilon}
\newcommand{\al}{\alpha}
\newcommand{\OUT}{\mathrm{OUT}}
\newcommand{\IN}{\mathrm{IN}}
\newcommand{\TV}{\mathcal{T}_V}
\newcommand{\Rc}{\mathcal{R}}
\newcommand{\TVd}{\mathcal{T}_{V,\ee}}
\newcommand{\GV}{\mathcal{G}_V}
\newcommand{\GVd}{\mathcal{G}_{V,\ee}}
\newcommand{\D}{\mathcal{D}}
\newcommand{\Sc}{\mathcal{S}}
\newcommand{\Gc}{\mathcal{G}}
\newcommand{\HH}{\mathcal{H}}
\newcommand{\Tt}{\mathcal{T}}
\newcommand{\MM}{\mathcal{M}}
\newcommand{\FF}{\mathcal{F}}
\newcommand{\Ec}{\mathcal{E}}
\newcommand{\Vc}{\mathcal{V}}
\newcommand{\WW}{\mathcal{W}}
\newcommand{\mm}{\mathcal{m}}
\newcommand{\J}{\mathcal{J}}
\newcommand{\Pc}{\mathcal{P}}
\newcommand{\V}{\mathcal{V}}
\newcommand{\Ic}{\mathcal{I}}
\newcommand{\Hc}{\mathscr{H}}
\newcommand{\Lc}{\mathcal{L}}
\newcommand{\Kc}{\mathcal{K}}
\newcommand{\Uc}{\mathcal{U}}
\newcommand{\Bb}{\mathcal{B}}
\newcommand{\Cc}{\mathcal{C}}
\newcommand{\Yc}{\mathcal{Y}}

\newcommand{\lo}{\longrightarrow}
\newcommand{\li}{\left}
\newcommand{\re}{\right}
\newcommand{\mi}{\,\,\big|\,\,}
\newcommand{\eq}{\Longleftrightarrow}
\newcommand{\cupp}{\mathop{\cup}\limits}
\newcommand{\capp}{\mathop{\cap}\limits}
\newcommand{\indeg}{\mathrm{indeg}}
\newcommand{\outdeg}{\mathrm{outdeg}}
\newcommand{\dist}{\mathrm{dist}}
\newcommand{\Os}{\mathrm{O}_{\mathrm{si}}}
\newcommand{\Oe}{\mathrm{O}_{\mathrm{el}}}
\newcommand{\Oee}{\mathrm{O}^0_{\mathrm{el}}}
\newcommand{\Ps}{\mathrm{P}_{\mathrm{si}}}
\newcommand{\Pe}{\mathrm{P}_{\mathrm{el}}}
\newcommand{\el}{\mathrm{el}}
\newcommand{\si}{\mathrm{si}}
\newcommand{\spann}{\mathrm{span}}
\newcommand{\head}{\mathrm{head}}
\newcommand{\tail}{\mathrm{tail}}
\newcommand{\crit}{\mathrm{crit}}
\newcommand{\link}{\mathrm{link}}
\newcommand{\ran}{\mathrm{ran}}
\newcommand{\im}{\mathrm{im}}
\newcommand{\rk}{\mathrm{rank}}
\newcommand{\Iso}{\mathrm{Iso}}
\newcommand{\Aut}{\mathrm{Aut}}
\newcommand{\rank}{\mathrm{rank}}

\newcommand{\co}{\color{red}}

\newtheorem{experiment}{Experiment}[section]
\newtheorem{remark}{Remark}[section]
 \newtheorem{problem}{Problem}

\begin{abstract}
The classical NP--hard \emph{feedback arc set problem} (FASP) and  \emph{feedback vertex set problem} (FVSP) ask for a minimum set of arcs $\ee \subseteq E$ or vertices $\nu \subseteq V$ whose removal
$G\setminus \ee$, $G\setminus \nu$ makes a given multi--digraph $G=(V,E)$ acyclic, respectively.
Though both problems are known to be APX--hard, approximation algorithms or proofs of inapproximability are unknown.
We propose a new $\mathcal{O}(|V||E|^4)$--heuristic for the directed FASP.
While a ratio of $r \approx 1.3606$ is known to be a lower bound for the APX--hardness, at least by empirical validation we achieve an approximation of $r \leq 2$.
The most relevant applications, such as \emph{circuit testing}, ask for solving the FASP on large sparse graphs, which can be done
efficiently within tight error bounds due to our approach. 

\end{abstract}

%



\maketitle

\section{Introduction}

Belonging to R.~M.~Karp's famous list of 21 NP--complete problems \cite{Karp:1972}, the FVSP \& FASP are of central interest in
theoretical computer science and beyond.

The most relevant applications  occur in \emph{electronic engineering} for \emph{designing processors or computer chips}. The chip design can be represented by a directed graph $G$, where the direction indicates the possible communication between the
chip components. Consistent testing or simulation of the signal process requires to consider sub-designs of feed-forward communication.
These sub-designs can be represented by acyclic subgraphs $G'\subseteq G$, which may be derived by solving the FASP. Especially, \emph{circuit testing} \cite{chip,gupta,VLSI,kunzmann,leiserson,global,unger}, including of
\emph{field programmable gate arrays (FPGAs)} \cite{FPGA2,FPGA,mars} rely on this approach.
Further problems and applications include \emph{efficient deadlock resolution in operating systems} \cite{logic,Silber},
\emph{minimum transversals of directed cuts} \cite{Lucchesi:1978} and \emph{general minimum multi-cuts} \cite{Even:1998}, \emph{ computational biology and
neuroscience} \cite{bao,greedy,Crick1998,Markov2014}. We recommend \cite{bang,marti} for exploring further relations to graph theoretical problems.
It is notable that the typical instances of the  mentioned applications are represented by graphs of \emph{large and sparse nature}.

While the undirected version of the FASP can be solved efficiently by computing  a maximum spanning tree
the undirected FVSP remains NP--complete. The only known (directed) instance classes possessing polynomial time solutions are planar or more general
weakly acyclic graphs \cite{Groetschel:1985}. Parameter tractable algorithms are given in \cite{Chen,FASP}. Further, by $L$--reduction of the \emph{minimum vertex cover problem} (MVCP) both problems are known to be APX--hard \cite{Karp:1972} and inapproximable beneath a ratio of
$r \approx 1.3606$  \cite{Dinur}, unless P=NP.
The undirected FVSP can be approximated within ratio $r=2$ \cite{bafna,becker} and is thereby APX--complete. The FASP on tournaments possesses a PTAS \cite{tour}. We recommend \cite{bang} for further studies.
That the directed FVSP \& FASP are approximation preserving $L$-reducible to each other is known due to
\cite{ausiello,crescenzi1995compendium,Even:1998,kann1992}. In our previous work \cite{FASP} we compactified these constructions and recapture them in Appendix \ref{Dual}.

The complementary problem of finding the \emph{maximum acyclic subgraph} is known to be MAX SNP complete and thereby approximable \cite{hassin,marti}. However,
this fact is not sufficient to approximate the directed FVSP \& FASP.
Though approximations of constrained versions of the problems \cite{Even:1998} were delivered, these approximations depend logarithmically on the number of cycles a graph possesses.
By reduction from the  Hamiltonian cycle problem,  counting all cycles is already a
\#P--hard counting problem \cite{Arora} and thereby the proposed approximation is not bounded constantly.

In \cite{Eades1993,marti,Saab} an excellent overview of heuristic solutions is given. Further, \cite{Eades1993} proposes the most common state of the art heuristic termed \emph{Greedy Removal} (GR).
Exact methods use ILP--solvers with modern formulations given in \cite{exact,sagemath} based on the results of \cite{Groetschel:1985,younger}. Heuristic and exact solutions for the weighted version are discussed in \cite{flood}.
However, for dense or large (sparse) graphs the 
ILP--approaches become sensitive to the NP--hardness of the FASP and require infeasible runtimes  while the
heuristic approach GR performs inaccurately. As we present in this article, our proposed heuristic solution can fill this gap and produce reasonable results.

\section{Theoretical Considerations}
In this section we provide the main graph theoretical concepts, which are required throughout the article.

\subsection{Preliminaries}

We address the feedback arc set problem in the most general setup. For this purpose, we introduce a non--classical definition of graphs as follows.

\begin{definition} \label{graph} Let $G=(V,E,\head,\tail)$ be a $4$--tuple, where $V,E$ are finite sets and $\head, \tail : E \lo V$ are some maps. We call the elements $v \in V$  \emph{vertices} and the elements $e \in E$  \emph{arcs} of $G$, while
$\head(e),\tail(e) \in V$ are called \emph{head} and \emph{tail} of the arc $e$. An arc $e$ with $\head(e)=\tail(e)$ is called a \emph{loop}.  In general, we call $G$ a \emph{multi--digraph}. The following cases are often relevant:
\begin{enumerate}
 \item[i)] $G$ is called a \emph{digraph} iff the map $H : E \lo V\times V$, with  $H(e)=(\tail(e),\head(e))$ is injective.
 \item[ii)] $G$ is called an \emph{undirected graph} iff $G$ is a digraph and for every $e \in E$ there is $f \in E\setminus\{e\}$ with $\head(e)=\tail(f)$, $\tail(e)=\head(f)$. In this case, we slightly simplify notation by
 shortly writing $e$ for the pair $e:=(e,f)\in E \times E$, which is then called an \emph{edge}. The notion of $\head,\tail$ can thereby be replaced by $\link: E \lo V$ with  $\link(e)=\head(e)\cup\tail(e)$.
\item[iii)] In the special case, were $E \subseteq V\times V$ the maps $\head,\tail$ are assumed to be canonically given by the relation of $E$, i.e., $\head((x,y))=y$, $\tail((x,y))=x$ for all $(x,y) \in E$.
\end{enumerate}
\end{definition}

One readily observes that, in the cases $i),ii)$, our definition coincides with the common understanding of graphs. In the general case of multi--digraphs, our definition has the advantage that though
\emph{multiple arcs} $e,f$ with $\head(e)=\head(f)$, $\tail(e)=\tail(f)$ are allowed $e,f$ are distinguished. Thus, $E$ is no multi--set, as it is assumed usually, but a simple set, simplifying our considerations.
For $e \in E$ we denote with $\vec F(e) = \li\{f \in E \mi \head(f)=\head(e), \tail(e)=\tail(f)\re\}$, $\cev F(e) = \big\{f \in E \mi \tail(f)=\head(e)$, $\head(f)=\tail(e)\big\}$,
 $F = \vec F(e) \cup \cev F(e)$ the set
of all parallel and anti--parallel arcs and their union, respectively.

Further, two arcs $e$ and $f$ are called \emph{consecutive} if
$\head(e)=\tail(f)$ and are called \emph{connected} if $\{\head(e),\tail(e)\}\cap
\{\head(f),\tail(f)\} \not =\emptyset$.
A \emph{directed path} $p=\{e_1,\dots,e_n\} \subseteq E$ of length $n \in \N$ from a vertex $u$ to a vertex $v$ is a list of
consecutive arcs $e_i \in E$, $i =1, \dots,n$ such that $u = \tail(e_1)$ and $v=\head(e_n)$. Thereby, repetition is allowed, i.e., $e_i =e_j$, $1\leq i<j\leq n$ is possible.

\begin{figure}[t!]
 \centering
 \vspace{-0.6cm}
\includegraphics[width=13cm]{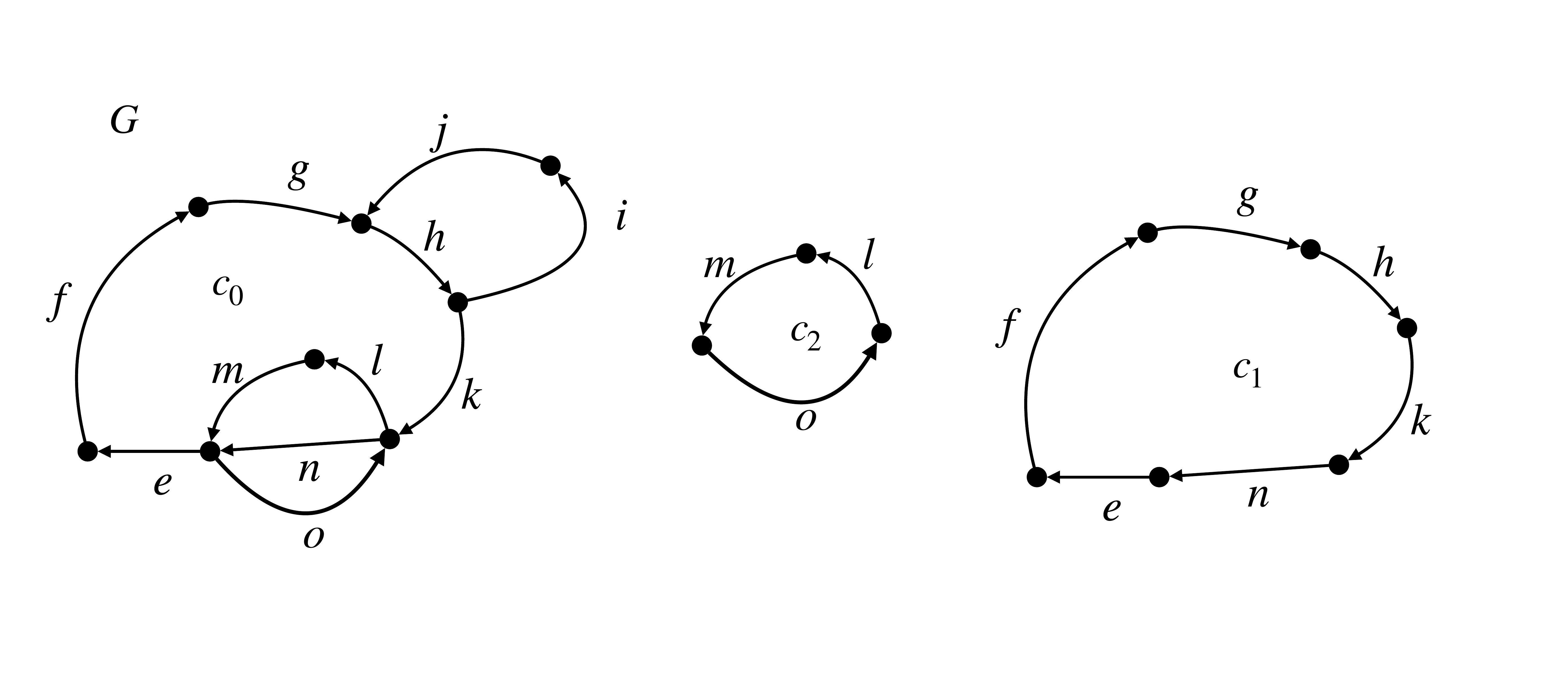}
  \vspace{-1.7cm}
 \caption{Elementary and simple cycles.}
 \label{fig:Ge}
\end{figure}

A directed  \emph{cycle} is a directed path $p$ from some vertex $v \in V$ to itself, which can also be a \emph{loop}. $O(G)$ shall denote the sets of all directed  cycles of $G$.
A cycle is called \emph{simple} or \emph{elementary} if every arc or vertex it contains is passed exactly once, respectively. Certainly, every cycle is given by passing through several elementary cycles. We denote with $\Oe(G) \subseteq O(G)$,
the set of all \emph{directed elementary cycles}.

\begin{example}
\label{exa:elementary}
Consider the graph $G$ in Figure~\ref{fig:Ge}. The cycle $c_0=\{e,f,g,h,k,l,m,o,n\}$
is a simple and non-elementary cycle, while the cycles $c_1=\{e,f,g,h,k,n\}$ and  $c_2=\{o,l,m\}$ are
elementary. Certainly, $c_0$ is given by passing through $c_1$ and $c_2$.
\end{example}

With  $G\setminus e$, $G\setminus v$ we denote the graphs obtained by deleting the arc $e$ or the vertex $v$ and all its connected arcs. Further,
$\Gc(\cdot)$, $\Ec(\cdot)$, $\Vc(\cdot)$ denote the
graph, the set of all arcs, and the set of all vertices
induced by a set of graphs, arcs and vertices.
By $\Pc(A)$ we denote the power set of a given  set $A$ of finite cardinality $|A|\in \N$.

\subsection{Problem Formulation}
In the following we formulate the classical optimization problems considered in this article.

\begin{problem}[FASP \& FVSP]\label{def:fasp}
Let $G =(V,E,\head,\tail)$ be  a multi--digraph and
$\omega: E\lo \R^+$ be an arc weight function.  Then the \emph{weighted FASP} is to
find a  set of arcs $ \ee \in \Pc(E)$  such that $G \setminus \ee$ is acyclic,
i.e., $O(G \setminus \ee) = \emptyset$ and
\begin{equation}
  \Omega_{G,\omega}(\ee):=\sum_{e \in \ee} \omega(e)
\end{equation}
is minimized. The weighted minimum feedback vertex set problem (FVSP) is given by considering a vertex weight function $\psi: V\lo \R^+$ and ask for a
set of vertices $ \nu \in \Pc(V)$  such that $G \setminus \nu$ is acyclic,
i.e., $O(G \setminus \nu) = \emptyset$ and
\begin{equation}
  \Psi_{G,\psi}(\nu):=\sum_{v \in \nu} \psi(v)
\end{equation}
is minimized.
We denote the set of solutions of the FASP \& FVSP with $\mathcal{F}_E(G,\omega), \mathcal{F}_V(G,\psi)$, respectively.

Further, we call $\ee \in \mathcal{F}_E(G,\omega)$  a \emph{minimum feedback arc set} and $\nu \in \mathcal{F}_V(G,\psi)$  a \emph{minimum feedback vertex set}
and denote with $\Omega(G,\omega)$, $\Psi(G,\psi)$ the \emph{minimum feedback arc/vertex length}.
If $\omega$ or $\psi$ are constant functions then we derive the unweighted versions of the FASP \& FVSP, respectively.
\end{problem}

\begin{remark}
 Note that, checking whether a graph is acyclic or not can be done by \emph{topological sorting} in $\Oc(|E|)$--time
\cite{cormen,kahn,tarjan_sort}. Further, every directed cycle is given by passing through several elementary cycles. Thus, the conditions $O(G) = \emptyset$ and $\Oe(G) = \emptyset$ are equivalent.
The FVSP \& FASP can be  also formulated in terms of the \emph{maximum linear ordering problem} see for instance \cite{exact,marti,younger}.
\end{remark}

\begin{figure}[t!]
 \centering
 \vspace{-0.25cm}
\includegraphics[width=12cm]{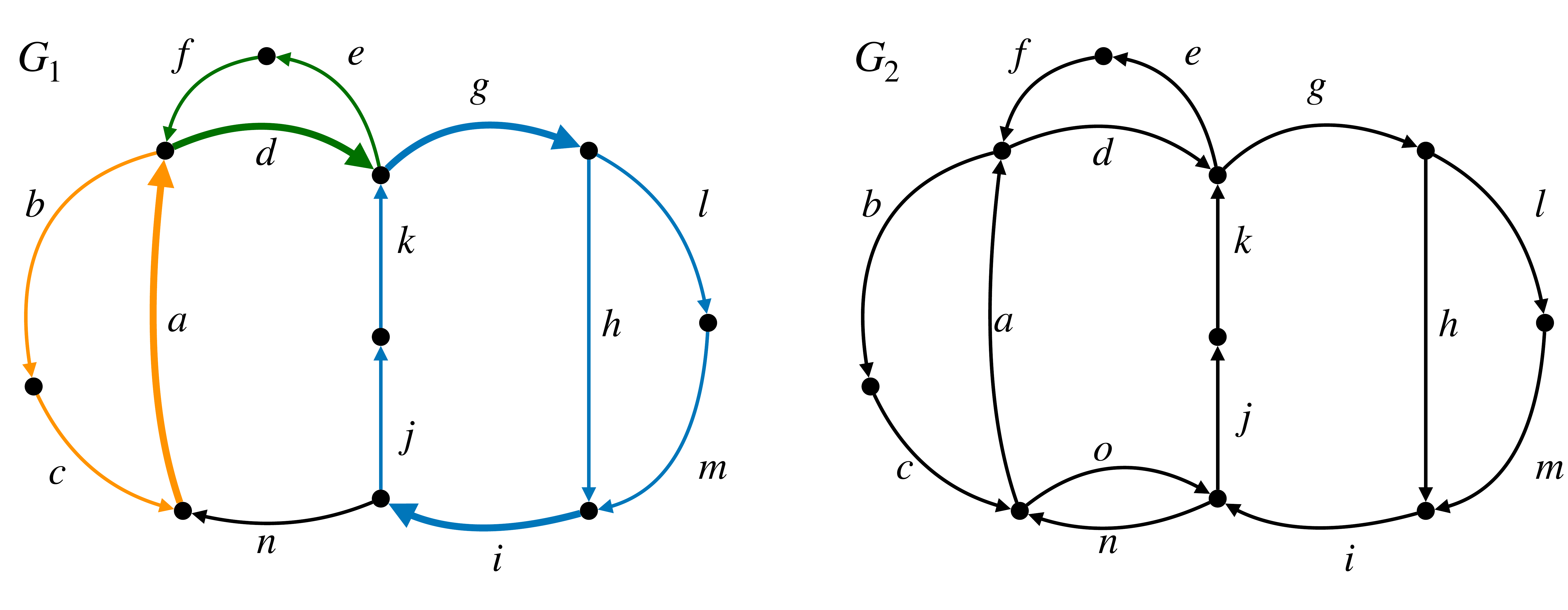}
\vspace{-0.5cm}
 \caption{$G_1$ with colored isolated cycles $\{a,b,c\}$, $\{d,e,f\}$, $\{g,h,i,j,k\}$, $\{g,l,m,i,j,k\}$ and $G_2$.}
 \label{fig:iso}
\end{figure}

\subsection{Isolated Cycles}\label{sec:Iso}

The complexity of an instance $G$ for the FVSP or FASP is certainly correlated to the structure of its cycles.
However, by reducing to the Hamiltonian cycle problem,
already counting all directed elementary cycles turns out to be a \#P-hard problem. This makes it hard to study the structure of $\Oe(G)$.
Here, we propose to use a technique developed in our previous article \cite{FASP} to overcome this issue.

\begin{definition}[cycle cover \& isolated cycles] \label{def:iso}
Let $G=(V,E, \head,\tail)$ be a multi--digraph and $e \in E$.  We call the subgraph
\begin{equation}
 G_e=(V_e,E_e, \head,\tail):= \Gc\li(\li\{ c\in \Oe(G) \mi \vec F(e) \cap  \Ec(c)\not = \emptyset \re\}\re)
\end{equation}
induced by all elementary cycles passing through $e$ or a parallel arcs $f \in \vec F(e)$ the \emph{cycle cover} of $e$.
Further, we denote with
$$I_e: =\Gc\li(\li\{ c \in \Oe(G) \mi  e \in \Ec(c) \,, \quad  c\cap c^\prime  = \emptyset \,, \,\,\, \text{for all}\,\,\, c^\prime
\in \Oe(G \setminus \vec F(e) \re\}\re)$$
the subgraph induced by all \emph{isolated cycles} passing through $e$, i.e., if $c$ is isolated then $c$ intersects with no cycle $c'$ passing not through $e$ or some parallel arc of $e$.
\end{definition}

\begin{remark}\label{rem:hierachy}
Note, that every loop is an isolated cycle. Further, the sets of isolated cycles possess a \emph{flat hierarchy} in the
following sense. If $e,f\in E$, with  $I_e\not = I_f$ then $I_e\cap I_f = \emptyset$.
Vice versa $I_e\cap I_f \not = \emptyset$  implies $I_e = I_f$ and further $G_e=G_f$.
\end{remark}

\begin{example}\label{ex:iso} Consider the graph $G_1$ in Figure~\ref{fig:iso}. Then $a,d,g,i$ are all arcs of $G_1$ with $I_x\not=\emptyset$, $x \in \{a,d,g,i\}$. The corresponding isolated cycles are colored.
Indeed, $I_g=I_i$ and $G_g=G_i$ holds according to Remark \ref{rem:hierachy}. Adding the arc $o$ to $G_1$ yields the graph $G_2$. Thereby,  all isolated cycles are connected with each other.
Thus,  $I_e=\emptyset$ for all arcs $e$ of $G_2$.
\end{example}

\pagebreak
\begin{theorem} \label{prop:Iso} Let $G=(V,E,\head,\tail)$ be a multi--digraph with arc weight $\omega : E \lo \R^+$.
\begin{enumerate}
\item[i)] There exist algorithms computing the subgraph $G_e$ in $\Oc(|E|^2+|V|)$ and $I_e$ in $\Oc(|E|+|V|)$.
 \item[ii)] If $\ee \in \mathcal{F}_E(G,\omega)$ is a minimum feedback arc set with  $e \in \ee$ then $\vec F(e) \subseteq \ee$.
 \item[iii)] If $I_e\not = \emptyset$  and  $\delta=\mathrm{mincut}(\head(e),\tail(e),I_e,\omega_{| I_e}) \in \Pc(E) $ be a minimum--$s$--$t$--cut with $s = \head(e)$, $t =\tail(e)$  w.r.t. $I_e, \omega_{|I_e}$ such that:
\begin{equation}\label{MC}
 \Omega_{G,\omega}(\delta) \geq \Omega_{G,\omega}(\vec F(e)) \,.
\end{equation}
Then there is $\ee \in \mathcal{F}_E(G,\omega)$ with $\vec F(e) \subseteq \ee$.
\item[iv)] Checking whether $I_e \not =\emptyset$ and \eqref{MC} holds can be done in $\Oc(|E||V|)$.
\end{enumerate}
\end{theorem}

\begin{proof} For better readability we recaptured this statement from our previous article \cite{FASP} in Appendix \ref{app:FASP}.
\end{proof}
Theorem \ref{prop:Iso} allows to localize optimal arc cuts $e \in E$ even in the weighted case, whenever isolated cycles with property $iii)$ exist.
We will use this circumstance  to propose an algorithm solving the FASP.

\section{The Algorithm}
The building block of the algorithm relies on applying Theorem \ref{prop:Iso} as presented below.  

\subsection{Building Block}

Given a multi--digraph $G=(V,E,\head,\tail)$ we formulate an algorithm termed ISO--CUT, which searches for arcs $e \in E$ which satisfy the assumption $iii)$ of Theorem \ref{prop:Iso}.
If such an arc $e \in E$ is located we store $e$ in a list $\ee$, consider $G=G\setminus e$ and continue the search until either the resulting graph $G$ is acyclic or no desired arc can be localized.
In any case, the stored arcs $\ee$ are an optimal subsolution for the FASP on $G$. A formal pseudo-code for the algorithm is given in Algorithm \ref{alg:ISO}.

\begin{algorithm}[t!]
\caption{ISO-CUT}\label{alg:ISO}
{\footnotesize
\begin{algorithmic}[1]
\Procedure{ISO-CUT}{$G,\omega$}\Comment{$\omega$ is an arc weight.}
        \State $\ee=\emptyset$, $\mathrm{iso}=1$
        \While{$G$ is not acyclic \& $\mathrm{iso}=1$} \Comment{Check by topological sorting in $\Oc(|E|)$.}
        \State $\mathrm{iso}=0$
 \For{ $e \in E$}
        \State Compute $I_e$ and $\delta=\mathrm{mincut}(\head(e),\tail(e),I_e,\omega_{| I_e})$
        \vspace{0.1cm}
        \If{$I_e\not= \emptyset$ \& $\Omega_{G,\omega}( \vec F(e)) \leq \Omega_{G,\omega}(\delta)$} \Comment{2nd condition redundant if $G$ is a digraph with $\omega \equiv 1$.}
        \State $G =G\setminus e$
        \State $\ee=\ee \cup \vec F(e)$
        \State $\mathrm{iso}=1$
        \EndIf
      \EndFor
      \EndWhile
      \State \textbf{return} $(G,\ee)$
\EndProcedure
\end{algorithmic}
}
\end{algorithm}

\begin{lemma}\label{lemma:Isocut} Let $G=(V,E,\head,\tail)$ be a multi--digraph with arc weight $\omega : E \lo \R^+$.
\begin{enumerate}
\item[i)] The algorithm ISO--CUT requires $\Oc(|E|^3+|V||E|^2)$ runtime to return an optimal subsolution $\ee\subseteq E$ of the FASP on $G$ and the remaining graph $G\setminus \ee$ in the unweighted case.
\item[ii)] The analogous return in the weighted case requires $\Oc(|V||E|^3)$ runtime.
\item[iii)] If $G\setminus \ee$ is acyclic then $\ee \in \mathcal{F}_E(G,\omega)$ is a minimum feedback arc set.
\end{enumerate}
\end{lemma}
\begin{proof} As one can verify readily Algorithm \ref{alg:ISO} contains 2 recursion over $E$ with line 6 being the bottleneck for each recursion.
The runtime estimation thereby follow directly from Theorem \ref{prop:Iso} $i)$ and $iv)$. Statement $iii)$ follows from Theorem \ref{prop:Iso} $iii)$.
\end{proof}
Note that due to Remark \ref{rem:hierachy} the algorithm ISO--CUT removes all loops from $G$.

\subsection{A Good Guess}
Though isolated cycles allow to localize optimal cuts
they do not need to exist at all, as the Example \ref{ex:iso} shows. Thus, 
the algorithm ISO--CUT might not return an acyclic graph. In this case  we have to develop a concept of a \emph{good guess} for cutting $G$ in a pseudo--optimal way until it possesses isolated cycles and 
thereby ISO--CUT can proceed. Our idea is based on the following fact.

\begin{proposition} Let $G=(V,E,\head,\tail)$ be a multi--digraph with arc weight $\omega : E \lo \R^+$,  $e \in E$, with $G_e\not = \emptyset$ and $I_e = \emptyset$.
Denote with $\delta= \mathrm{mincut}(\head(e),\tail(e),G_e,\omega_{| G_e})$ a minimum--$s$--$t$--cut
and with $\ee \in \mathcal{F}_E(G,\omega)$ a minimum feedback arc set, while $\ee' = \ee \cap \Ec(G_e)$ shall denote its restriction to $G_e$.
If
\begin{equation}\label{ratio}
 \Omega_{G,\omega}(\delta) - \Omega_{G,\omega}(\vec F(e)))  >  \Omega(G\setminus \vec F(e),\omega) -\Omega(G\setminus \ee',\omega)
\end{equation}
then $\vec F(e) \subseteq  \ee$.
\end{proposition}

\begin{proof}  Assume $e \not \in \ee$ then due to Theorem \ref{prop:Iso} $ii)$ there holds  $\ee \cap \vec F(e) = \emptyset$. Since $G_e \setminus \ee'$ is acyclic and $\delta$ is a minimim $s$--$t$--cut we obtain 
$\Omega_{G,\omega}(\ee') \geq   \Omega_{G,\omega}(\delta)$. 
Hence, by rewriting \eqref{ratio} we can estimate  
$$ \Omega_{G,\omega}(\vec F(e))) +\Omega(G\setminus \vec F(e),\omega)  < \Omega_{G,\omega}(\delta) + \Omega(G\setminus \ee',\omega) \leq  \Omega_{G,\omega}(\ee') +  \Omega(G\setminus \ee',\omega) =  \Omega(G,\omega)\,,$$
which yields a contradiction and thereby proves the claim.
\end{proof}
Certainly, the right hand side of \eqref{ratio} is hard to compute or even to estimate. Intuitively, one could  guess that the larger the left hand side becomes the more likely it is that the inequality in \eqref{ratio} holds.
This intuition is the basic idea of our concept of a \emph{good guess}.

However, maximizing the left hand side of \eqref{ratio} is too costly for a heuristic guess. Therefore, we restrict our considerations to one cycle $c \in \Oe(G)$ and all arcs $e_1,\dots,e_n \in \Ec(c)$
with $\Gc(c) \subsetneq G_{e_i}$, $i=1,\dots,n$ cutting more cycles than $c$. Now we choose
\begin{equation}\label{Good}
 \text{GOOD--GUESS}(G,\omega,c)  = \mathrm{argmax}_{e_i,i=1,\dots,n}\Big(\mathrm{mincut}(\head(e_i),\tail(e_i)) - \Omega_{G,\omega}(\vec F(e_i)\Big)
\end{equation}
as an arc with the most expansive minimum $s$--$t$--cut to be the one to cut.  By combining \cite{KRT} and \cite{orlin} computing minimum $s$--$t$--cuts requires $\Oc(|E||V|)$. Consequently, this heuristical
decision can be made efficiently in $\Oc(|c||E||V|)$.

\begin{algorithm}[t!]
\caption{TIGHT--CUT}\label{alg:tight}
{\footnotesize
\begin{algorithmic}[1]
\Procedure{TIGHT--CUT}{$G,\omega)$}\Comment{$\omega$ is an arc weight.}
\State $\ee =\emptyset$, $\delta =\emptyset$
\For{ $i =1,\dots |E|$}
        \State $(G,\ee')=\text{ISO--CUT}(G,\omega)$
        \State $\ee = \ee \cup \ee'$
        \If{$G$ is acyclic}  \Comment{Can be checked by topological sorting in $\Oc(|E|)$}
        \State \textbf{break}
        \Else
        \State Choose $c \in Oe(G)$
        \State  $h=$GOOD--GUESS$(G,\omega,c)$
        \State  $\delta= \delta\cup\{h\}$
        \State  $G = G\setminus h$
        \EndIf
\EndFor
      \State \textbf{return} $(\ee \cup \delta)$, $\Omega_{G,\omega}(\delta)$
\EndProcedure
\end{algorithmic}
}
\end{algorithm}

\subsection{The Global Approach}
Now we combine the algorithms ISO--CUT and GOOD--GUESS to yield an algorithm termed TIGHT-CUT computing feedback arc sets formalized in Algorithm \ref{alg:tight}.

\begin{proposition}  Let $G=(V,E,\head,\tail)$ be a multi--digraph with arc weight $\omega : E \lo \R^+$ and vertex weight $\nu : V\lo \R^+$.
\begin{enumerate}
 \item[i)] The algorithm TIGHT--CUT proposes a feedback arc set $\ee\cup \delta \subseteq E$ in $\Oc(|V||E|^3 + |E|^4)$ in the unweighted case.
  \item[ii)] In the weighted case the analogous return requires  $\Oc(|V||E|^4)$ runtime.
 \item[iii)] The algorithm TIGHT--CUT can be adapted to proposes a feedback vertex  set $\nu\subseteq E$ in \linebreak $\Oc(|E|(\Delta(G)|V|)^3)$ in the unweighted case and
 $\Oc(|E|(\Delta(G)|V|)^4)$ in the weighted case, where $\Delta(G)$ denotes the maximum degree of $G$.
\end{enumerate}

\label{prop:half}
\end{proposition}
\begin{proof} Obviously TIGHT--CUT runs once through all arcs $E$ in the worst case. Due to the  notion of our GOOD--GUESS \eqref{Good} the bottleneck is thereby ISO--CUT. Thus, due to Lemma \ref{lemma:Isocut} we obtain $i)$, $ii)$.
Now $iii)$ is a consequence of an existing approximation preserving $L$--reduction from the FVSP to the FASP relying on Definition \ref{def:dualG} and Proposition \ref{Lredu}.
\end{proof}

\section{Validation \& Benchmarking}

To speed up the heuristic TIGHT--CUT we formulated a relaxed version, which we implemented in C++. The relaxation relies on
weakening condition $iii)$ in Theorem \ref{prop:Iso} by a notion of \emph{almost isolated cycles}. Further explanations and a pseudo--code are given in Appendix~\ref{relax} and Algorithm~\ref{alg:tightR}.
All benchmarks were run on a single CPU core on a machine with
CPUs: $2\times$ Intel(R) Xeon(R) E5-2660 v3 @ 2.60GHz;
Memory: 128GB;
OS: Ubuntu 16.04.6 LTS using compiler: GCC 9.2.1.%

The C++ code, as well as all benchmark datasets used in this work, are publicly available at \url{https://git.mpi-cbg.de/mosaic/FaspHeuristic}.

\begin{figure}[tbp]
 \centering
 \vspace{-0.25cm}
\includegraphics{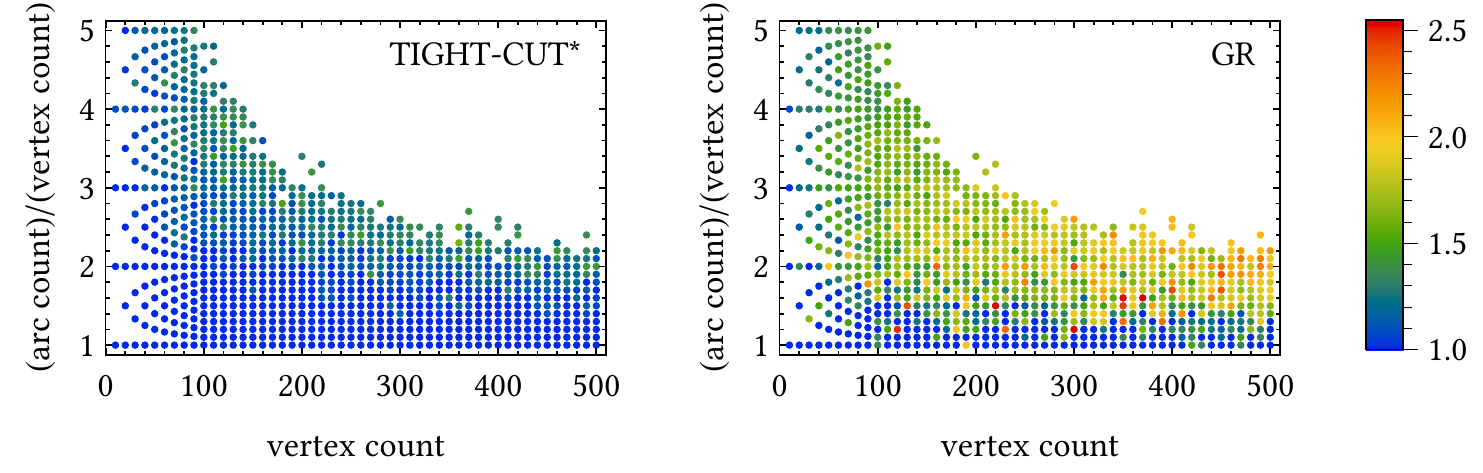}\\
\vspace{0.2cm}
\includegraphics{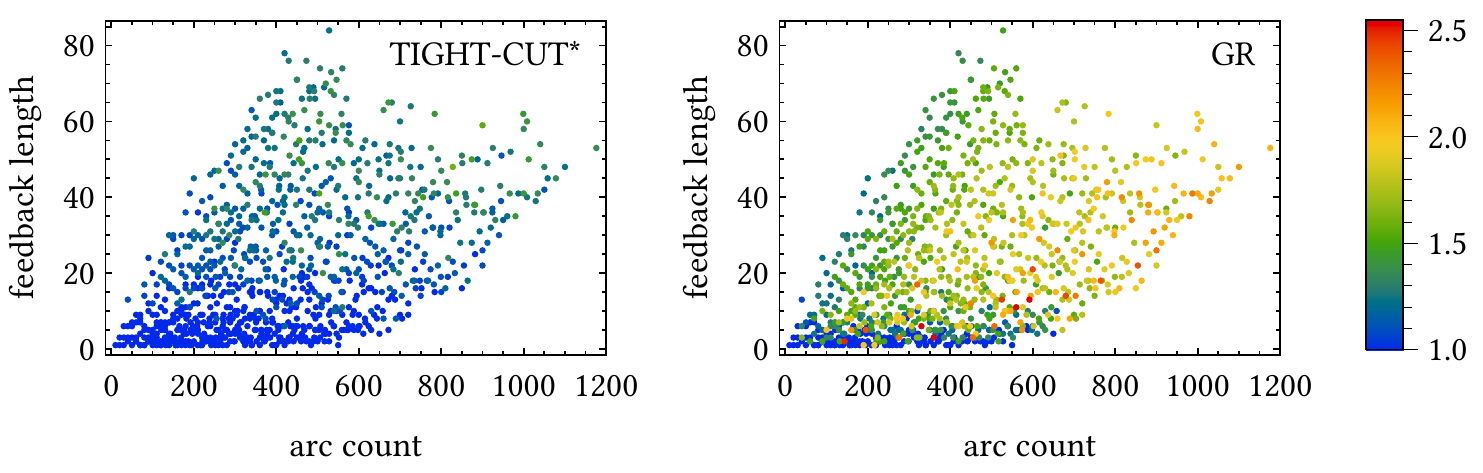}
  \vspace{-0.8cm}
 \caption{Approximation ratios of TIGHT--CUT* (left) and GR (right) for random graphs.}
\label{fig:ATG}
\end{figure}

The following implementations were used for the experiments:
\begin{enumerate}
 \item[I)] An exact integer linear programming based approach implemented as the \texttt{feedback\_edge\_set} function from \emph{SageMath 8.9} \cite{sagemath} with iterative constraint generation, termed  \emph{EM}.
 \item[II)] The greedy removal approach from  \cite{Eades1993}, termed  \emph{GR}, imported from the \emph{igraph library} \cite{IGR,IGraphM}.
 \item[III)] The relaxed version of TIGHT--CUT, termed \emph{TIGHT--CUT*}, presented in Appendix \ref{relax} with settings $n=3$, $N=20$.
\end{enumerate}

EM is similar to the approach from \cite{exact} and iteratively increases the \emph{cycle matrix} required for the optimization. Thereby, a sequence  $\ee_1,\dots,\ee_n$, $n \in \N$
of optimal subsolutions is generated with $\ee_n \in \mathcal{F}_E(G,\omega)$, $\omega\equiv 1$ being a global solution for $G$. Indeed, the method can not handle the weighted case.
We chose the GLPK back end for SageMath's integer programming solver, which we found to perform significantly better than COIN-OR's CBC or Gurobi.

\begin{figure}[t]
 \centering
\includegraphics{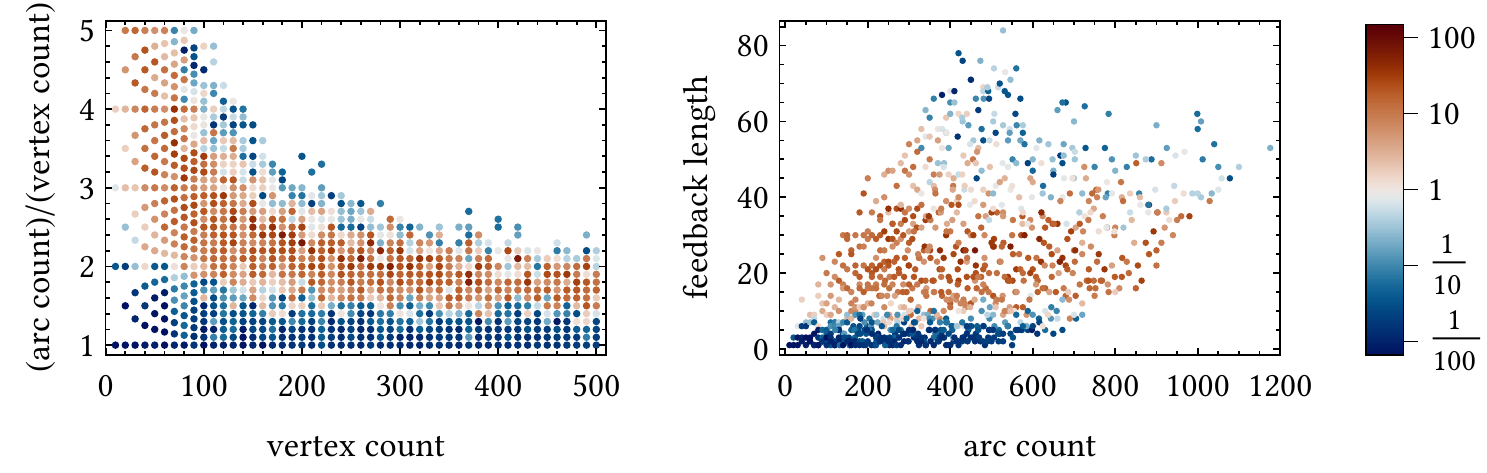}
  \vspace{-0.5cm}
 \caption{Runtime ratios of TIGHT--CUT* / EM plotted against $|E|/|V|$ (left) and feedback length (right).}
 \label{fig:TI}
\end{figure}

\subsection{Synthetic Instances}
In order to compare approximation ratios and runtimes we generated the following instance classes:
\begin{enumerate}
 \item[i)] We used the  Erd{\H{o}}s--R{\'e}nyi model in order to generate random digraphs  \cite{model}.
 \item[ii)] We uniform randomly chose a direction for every edge in a complete undirected graph $K_n$ in order to generate tournaments of size $n \in \N$.
 \item[iii)] We generated random maximal planar digraphs, then considered uniform perturbations of planarity by randomly \emph{rewiring}, i.e.~removing and re-inserting, a fraction $0 \leq p\leq 1$ of arcs.
 This construction is similar to the Watts--Strogatz \emph{small-world model} \cite{Watts1998}.
\item[iv)] We followed \cite{Saab} in order to generate large digraphs $G$ of known feedback arc length. Adaptions to treat the weighted case were made.
\end{enumerate}

\begin{experiment} In total we generated  1869 random digraphs. Figure~\ref{fig:ATG} shows the approximation ratios obtained by  TIGHT--CUT* and GR on 967 out of these 1869 graphs plotted
once against $|E|/|V|$ and once against the exact minimum feedback arc length.
The exact feedback length was determined by EM whose runtime ratio w.r.t. TIGHT--CUT* is plotted in Figure~\ref{fig:TI}. Thereby, the empty region in the left panel reflects the 902 instances
which EM could not process within EM--time--out $=30$ min.

Figure~\ref{fig:ATG} validates that TIGHT--CUT* approximates the FASP beneath a ratio of at most $1.6$ by being much tighter in most of the cases. On the other hand, GR reaches ratios up to $2.5$.
The parameter tractable algorithm of \cite{Chen} indicates that the feedback length reflects the complexity of a given instance. However, the accuracy of GR decreases quickly
with increasing graph size and $|E|/|V|$ regardless of the feedback length.
In contrast, the accuracy behavior of TIGHT--CUT* reflects that circumstance. Whatsoever, TIGHT--CUT* performs significantly better than GR. Especially, when approaching the time--out--region of EM,
the approximation ratios of TIGHT--CUT* remain small. Thus, for digraphs located above the red region in Figure~\ref{fig:TI} the plot in  Figure~\ref{fig:weight} shows that TIGHT--CUT* is up to $100$--times faster than EM.
Even though TIGHT--CUT*  requires up to $3$ min and GR runs beneath $1$ sec, on these graphs,
TIGHT--CUT*  is the only approach producing reasonable results.
 \end{experiment}

\begin{figure}[t]
\centering
\includegraphics{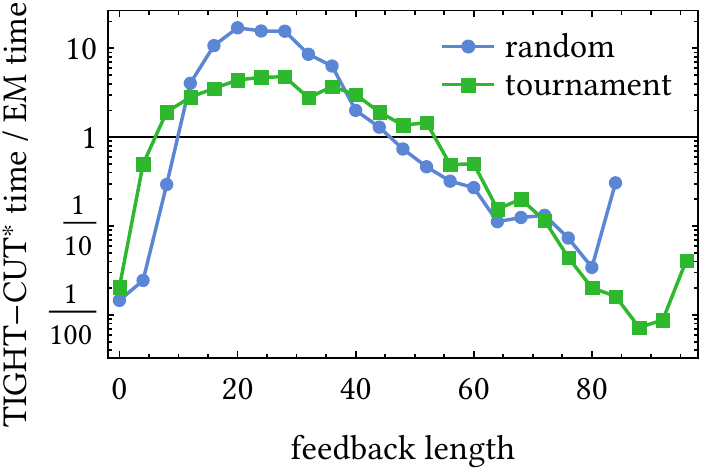}
\hspace{0.2cm}
\includegraphics{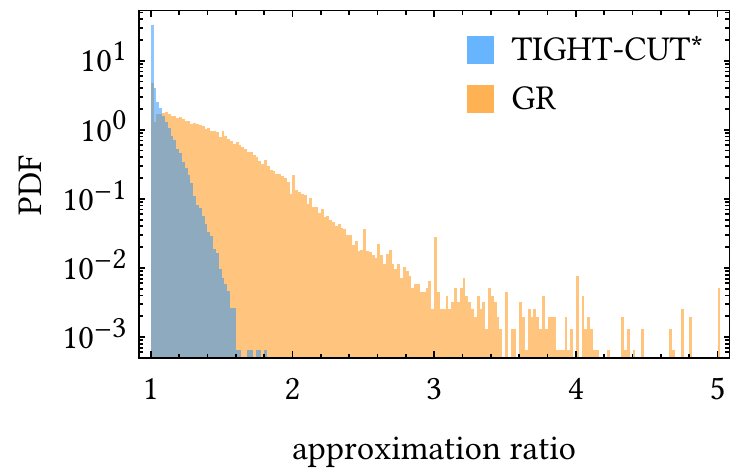}
  \vspace{-0.6cm}
 \caption{Runtime ratios of TIGHT--CUT* / EM (left) and distribution of TIGHT--CUT* and GR on weighted digraphs with logarithmic $y$--scale (right).}
\label{fig:weight}
\end{figure}
\begin{figure}[t]
 \centering
 \includegraphics{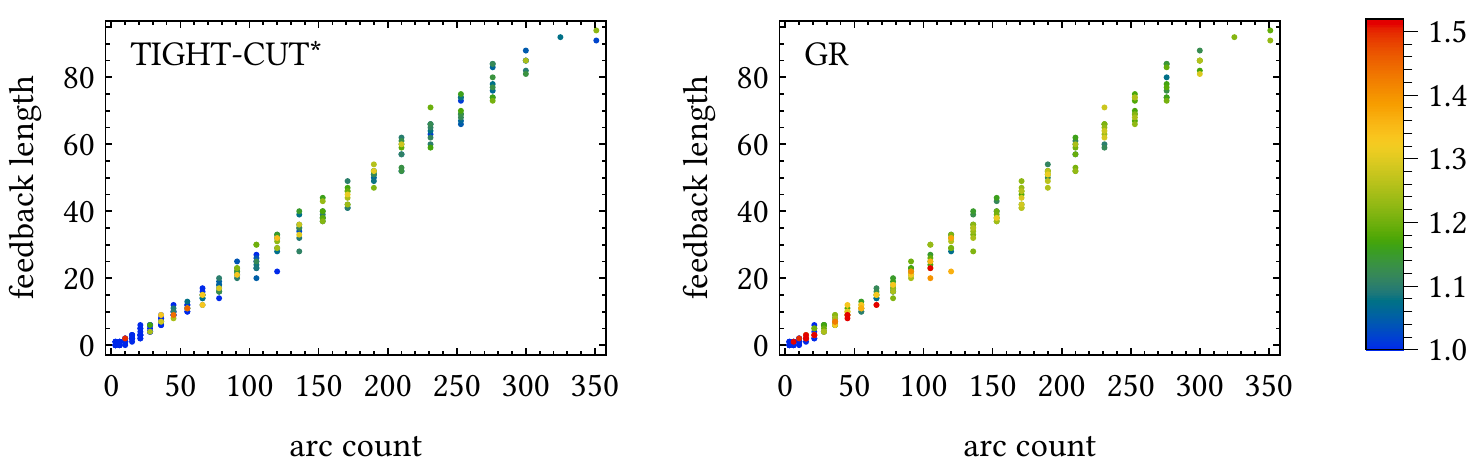}
 \includegraphics{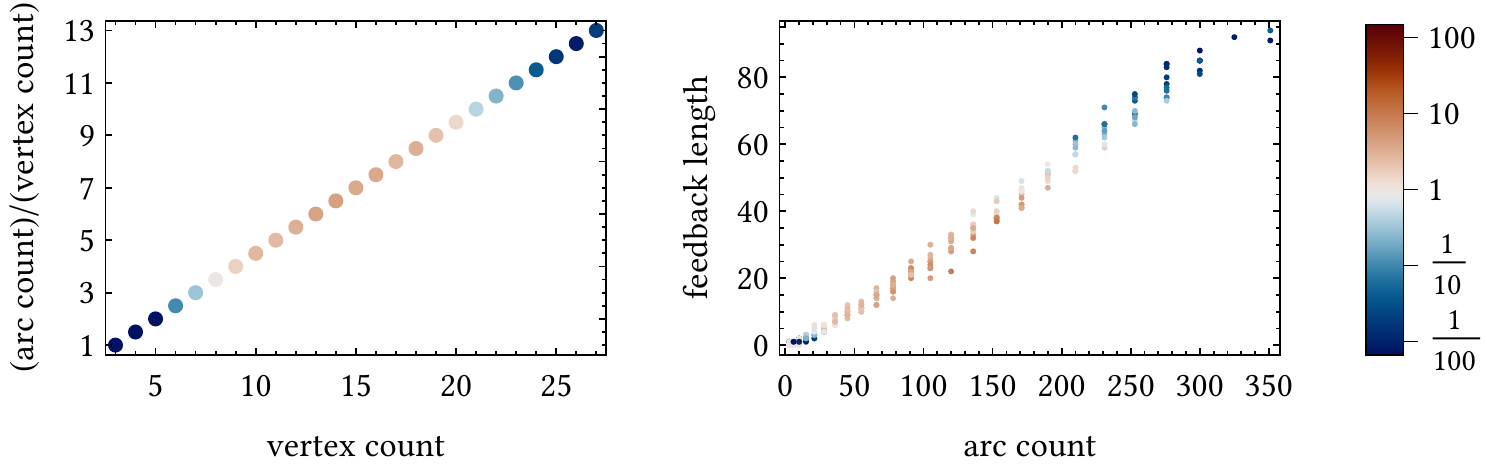}
  \vspace{-0.25cm}
 \caption{Approximation ratios (above) of TIGHT--CUT* (left) and GR (right)  and runtime ratios (below) of TIGHT--CUT*/EM  vs vertex count (left) and arc count  (right) on tournaments.}
\label{fig:tour}
\end{figure}

\begin{figure}[t]
 \centering
\includegraphics{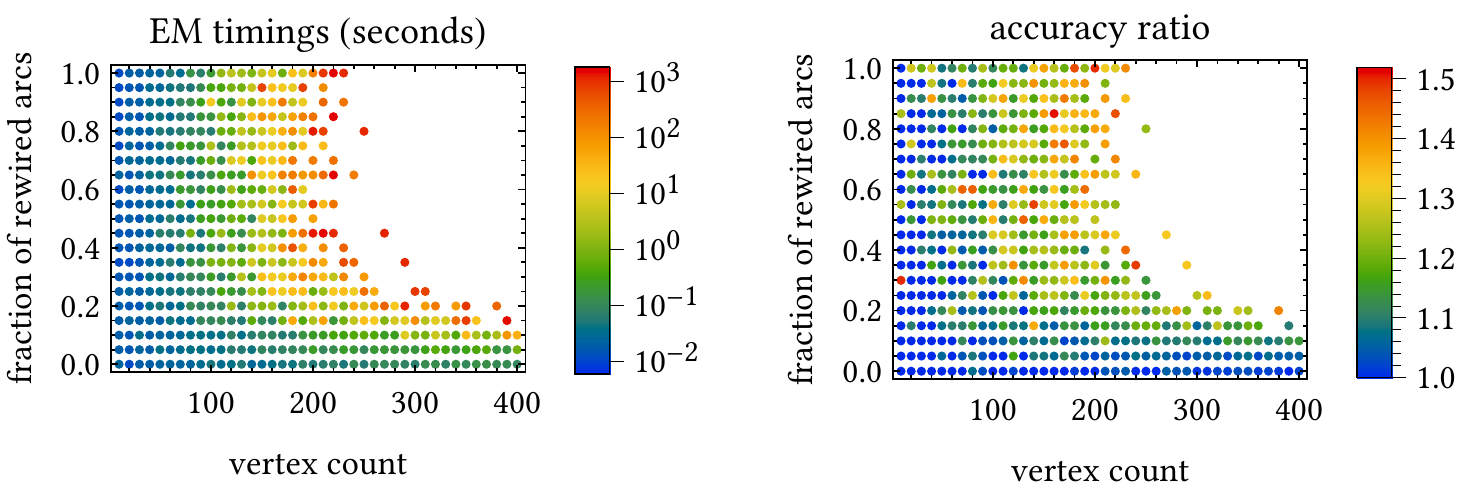}
  \vspace{-0.25cm}
 \caption{Runtimes of EM for perturbed planar graphs (left) and approximation ratio of TIGHT--CUT* (right).}
 \label{fig:world}
\end{figure}

\begin{experiment} We focus our considerations to tournaments. In Figure~\ref{fig:tour} the results for $|V|=1,\dots,27$ with 10 instances for each size are shown. $|V|=27$ is thereby the maximum size for EM not running into
time--out $=30$ min.
GR seems to perform only slightly worse than TIGHT--CUT*.
However, the feedback arc length for tournaments
averages about $25\%$ of its high arc count $|E| = (|V|^2 + |V|)/2$. Thus, the improvement in accuracy TIGHT--CUT* gains compared to GR is as significant. Again the NP-hardness of the FASP
becomes visible for the runtime ratios in Figure~\ref{fig:tour} and Figure~\ref{fig:weight} (left). As expected the feedback arc length is correlated to the complexity of the cycle structure of $G$.
Regardless of the type of the graphs we thereby reach intractable instances for EM beyond a feedback arc length of $\Omega(G,\omega)\geq 50$. Thus, for instances allowing $\Omega(G,\omega)\geq 50$,
EM might run into time out while TIGHT--CUT* processes them efficiently.
\end{experiment}

\begin{experiment}
  Since EM can not handle the weighted FASP we adapted the method of \cite{Saab} to generate 77700 weighted multi--digraphs $(G,\omega)$ of integer weights $\omega(e) = 1 \sim 10$ with known feedback arc length
  $\Omega(G,\omega) = 1\sim287$
  and sizes from $|V|=100\sim 500$ and $|E|/|V| =1.5\sim5$.
  Figure~\ref{fig:weight} (right) illustrates the results. To merge the ratio distributions of GR and TIGHT--CUT* on one plot we chose a logarithmic scaling for the $y$--axis.
  Indeed, TIGHT-CUT* approximates the FASP beneath a ratio of $2$ and solves more than $50\%$ exactly and $95\%$ beneath a ratio of $1.18$. In contrast GR is spread over ratios from $1$ to $5$
  producing exact solutions only for $8.8\%$ and $95\%$ beneath a ratio of $1.96$. Thus, though GR runs beneath $1$ sec and the runtimes of TIGHT--CUT* vary from seconds to $3$ minutes,
  this accuracy improvement justifies the larger amount of time.
\end{experiment}

\begin{experiment} In Figure~\ref{fig:world} the EM runtimes and  TIGHT--CUT* approximation ratios for 541 small--world (perturbed planar digraphs) are plotted. As one can observe already for small perturbations a similar behavior as for random graphs occurs.
In applications one can rarely guarantee planarity. At best, one can hope for planar--like instances. Consequently, real-world
instances, hinder the efficiency of ILP--Solvers on planar graphs to come into effect. Therefore, TIGHT--CUT* is an alternative to EM worth considering even for planar--like graphs.
\end{experiment}

\begin{figure}[t]
\centering
\includegraphics{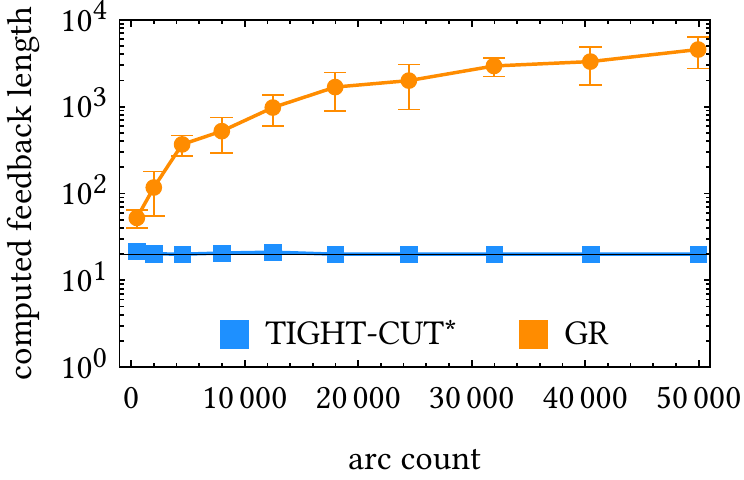}
 \vspace{-0.25cm}
 \caption{Accuracy for  TIGHT--CUT* and GR on very large digraphs of density $|E| / \bigl(|V|(|V|-1)\bigr) = 5\%$. The minimum feedback arc length of each instance is 20.}
\label{fig:large}
\end{figure}%

\begin{experiment}We measured the accuracy behavior of GR and TIGHT--CUT* in the time--out region of EM. Therefore, we generated very large unweighted digraphs with known feedback arc length
$\Omega(G,\omega) =20$ and varying vertex size $|V|=100,200,\dots,1000$, $10$ instances for each size, with  $5\%$ density, i.e., $|E|/\big(|V|(|V|-1)\big) = 5\%$.
In Figure~\ref{fig:large} the computed feedback lengths of both approaches are plotted with error bars indicating the standard deviation.
While TIGHT-CUT* delivers almost exact solutions GR is infeasible for these large graphs.
\end{experiment}

\begin{table}[t!]
\begin{tabular}{rrrrr}
 vertices & arcs & $\Omega(G,\omega)$ &  TIGHT--CUT* & approx. ratio \\
 \hline
 100 & 990 & 200 & 321 & 1.61 \\
 100 & 990 & 200 & 279 & 1.40 \\
 200 & 3980 & 200 & 280 & 1.40 \\
 500 & 1500 & 200 & 334 & 1.67 \\
 501 & 1501 & 200 & 345 & 1.73 \\
\end{tabular}
\caption{Examples of high feedback arc length for validating the ratio approximation.}
\label{TE}
\end{table}

\begin{table}[t!]
\begin{tabular}{r|rrrrr}
circuit name & vertices & arcs & $\Omega(G,\omega)$ & TIGHT--CUT* & GR \\
\hline
s27 & 55 & 87 & 2 & 2 & 2 \\
s208 & 83 & 119 & 5 & 5 & 5 \\
s420 & 104 & 178 & 1 & 1 & 1 \\
mm4a & 170 & 454 & 8 & 8 & 16 \\
s382 & 273 & 438 & 15 & 15 & 29 \\
s344 & 274 & 388 & 15 & 15 & 23 \\
s349 & 278 & 395 & 15 & 15 & 24 \\
s400 & 287 & 462 & 15 & 15 & 28 \\
s526n & 292 & 560 & 21 & 21 & 29 \\
mult16a & 293 & 582 & 16 & 16 & 23 \\
s444 & 315 & 503 & 15 & 15 & 20 \\
s526 & 318 & 576 & 21 & 21 & 31 \\
mult16b & 333 & 545 & 15 & 15 & 22 \\
s641 & 477 & 612 & 11 & 11 & 16 \\
s713 & 515 & 688 & 11 & 11 & 16 \\
mult32a & 565 & 1142 & 32 & 32 & 45 \\
mm9a & 631 & 1182 & 27 & 27 & 29 \\
s838 & 665 & 941 & 32 & 32 & 37 \\
s953 & 730 & 1090 & 6 & 6 & 11 \\
mm9b & 777 & 1452 & 26 & 27 & 31 \\
s1423 & 916 & 1448 & 71 & 71 & 112 \\
sbc & 1147 & 1791 & 17 & 17 & 21 \\
ecc & 1618 & 2843 & 115 & 115 & 137 \\
phase decoder & 1671 & 3379 & 55 & 55 & 64 \\
daio receiver & 1942 & 3749 & 83 & 83 & 123 \\
mm30a & 2059 & 3912 & 60 & 60 & 62 \\
parker1986 & 2795 & 5021 & 178 & 178 & 313 \\
s5378 & 3076 & 4589 & 30 & 30 & 75 \\
s9234 & 3083 & 4298 & 90 & 91 & 163 \\
bigkey & 3661 & 12206 & 224 & 224 & 224 \\
dsip & 4079 & 6602 & --- & 153 & 165 \\
s38584 & 20349 & 34562 & 1080 & 1080 & 1601 \\
s38417 & 24255 & 34876 & 1022 & 1022 & 1638 \\
\hline
ibm01 & 12752 & 36048 & --- & 1761 & 3254 \\
ibm02 & 19601 & 57753 & --- & 3820 & 5726 \\
ibm05 & 29347 & 98793 & 4769 & 4769 & 5979
\end{tabular}
\caption{Feedback arc set size results for graphs generated from the ISCAS and ISPD98 circuit benchmarks. TIGHT--CUT* is exact except in ``mm9b'' and ``s9234'', with failure $1$. For ``dsip'', no exact solution is available. Results for 
``ibm01'',``ibm02'', ``ibm05'' are given in the last 3 lines. }
\label{tab:iscas}
\end{table}

\begin{experiment} We generated a few very large and dense unweighted graphs with large feedback length $\Omega(G,\omega)=200$.
The results are listed in Table \ref{TE} and validate that TIGHT--CUT* approximates the FASP beneath a ratio of $2$.
\end{experiment}

\pagebreak 
\subsection{Real-world datasets}

Feedback problems find applications in \emph{circuit testing}, as efficient testing requires the elimination of feedback cycles \cite{chip,gupta,VLSI,kunzmann,leiserson,global,unger}.
Here we consider graphs generated from the ISCAS circuit testing datasets, made available in \cite{Dasdan2004} and at \url{https://github.com/alidasdan/graph-benchmarks}.
The results are summarized in Table~\ref{tab:iscas}. All examples were solved in runtime comparable to that of EM, except for ``dsip'', which could not be solved by EM within a computation time of 1 day,
and was solved by TIGHT--CUT* in 10 minutes. The runtime of the other examples ranged from milliseconds to minutes.

We also considered circuits from the ISPD98 benchmark \cite{Alpert1998,Dasdan2004}. These graphs are much larger, with arc counts ranging from $36\,000$ to $670\,000$.
The exact solution could only be obtained for one graph (``ibm05'') by EM with a runtime limit of 1 day. Thereby,  EM took 30 minutes and TIGHT--CUT* obtained an exact solution for ``ibm05'' in 2 minutes. 
Without limiting time--out, TIGHT--CUT* could process ``ibm01'', ``ibm02'' in 6 and 13 days, respectively.
The resulting feedback sizes are $1.25 \sim 1.85$ times smaller than the solutions proposed by GR, see again Table~\ref{tab:iscas}. 

The accuracy improvement gained by TIGHT--CUT* makes circuit testing much more efficient and robust for these graphs. Therefore, we aim  to speed up our implementation such that runtimes under 1 day can be achieved for the
ISPD98 instances. How these aims may become achievable and other remaining issues can be resolved is discussed in the final section.

\section{Conclusion}

We presented a new $\Oc(|V||E|^4)$--heuristic termed TIGHT--CUT of the FASP which is adaptable for the FVSP in $\Oc(|E|(\Delta(G)|V|)^4)$ processing even weighted versions of the problems.
At least by validation the ratio of the implemented relaxation TIGHT--CUT* is shown to be bounded by $2$ (in the unweighted case) and is much smaller for most of the considered instances.
Though we followed several ideas we can not deliver a proof of the APX--completeness for the directed FVSP \& FASP
at this time.
Nevertheless, we are optimistic that a deeper understanding of isolated cycles may provide a path for proving the ratio $r$ to be bounded by $2$ for all possible instances.
In any case, by Proposition \ref{Lredu} the directed FVSP \& FASP can be $L$-reduced to each other. Hence, either both problems are APX--complete or none of them.

Regardless of these theoretical questions, validation and benchmarking with the heuristic GR \cite{Eades1993} and the ILP--method EM from \cite{sagemath} demonstrated the high--quality performance of TIGHT--CUT*
even in the weighted case. Though of runtime complexity $\Oc(|V||E|^4)$, real-world instances, such as the graphs from the ISCAS circuit benchmarks, which are of large and sparse nature, can be solved efficiently within tight error bounds.

While our current implementation of TIGHT--CUT* can only solve a few of the instances from the ISPD98 circuit dataset, it provides a great accuracy improvement over GR. In light of this fact, we consider it worthwhile to spend further resources on improving the implementation. Runtime improvements of TIGHT--CUT* are certainly possible by parallelization, 
and by using more efficient implementations of subroutines, e.g.\ using a dynamic decremental computation of strongly-connected components \cite{acki2013}, and  using the improved minimum--$s$--$t$--cut algorithms 
from \cite{Goldberg2014}.
A fast implementation of the $L$--reduction from the FVSP to the FASP is in progress allowing to solve the FVSP by TIGHT--CUT* with the same accuracy in similar time.

We hope that many of the applications, even those which are not mentioned here, will benefit from our approach.

\begin{acks}
We thank Ivo F.\ Sbalzarini and Christian L.\ M\"uller for inspiring discussions and suggestions.
\end{acks}

\appendix
\section{Relaxed Version of TIGHT--CUT}
\label{relax}

In order to make the algorithm ISO--CUT faster and more effective we propose the following relaxation within TIGHT--CUT.

\begin{algorithm}[t!]
\caption{TIGHT--CUT*}\label{alg:tightR}
{\footnotesize
\begin{algorithmic}[1]
\Procedure{TIGHT--CUT*}{$G,\omega,n,N$}\Comment{$\omega$ is an arc weight, $n,N \in \N$.}
\State $\ee =\emptyset$, $\delta =\emptyset$
\For{ $i =1,\dots |E|$}
        \State $(G,\ee')=\text{ISO--CUT}(G,\omega)$
        \State $\ee = \ee \cup \ee'$
        \If{$G$ is acyclic}  \Comment{Can be checked by topological sorting in $\Oc(|E|)$}
        \State \textbf{break}
        \Else
        \State Choose $\mu_1,\dots,\mu_{N} \subseteq E$ with $|\mu_i|=n$ uniformly randomly.
        \State $(H_i,\varrho_i)=\text{ISO--CUT}(G\setminus \mu_i)$
        \State $R=\cup_{i=1}^{N}\varrho_{i,1}$ \Comment{$\varrho_{i,1}$ is the first arc cut by ISO--CUT.}
        \If{$R \not = \emptyset$}
        \State $f = \mathrm{argmax}_{e \in R}\big|\{  e=\varrho_{i,1}\}_{1\leq i\leq N}\big|$  \Comment{$f$ is a good choice in most of the $H_i$.}
        \State $\ee = \ee \cup\{f\}$
        \State $G = G \setminus f$
        \Else
        \State  Choose $c \in \Oe(G)$
        \State  $h=$GOOD--GUESS$(G,\omega,K)$
        \State  $\delta= \delta\cup\{h\}$
        \EndIf
        \State  $G = G\setminus h$
        \EndIf
\EndFor
      \State \textbf{return} $(\ee \cup \delta)$, $\Omega_{G,\omega}(\delta)$
\EndProcedure
\end{algorithmic}
}
\end{algorithm}

\begin{definition}[almost isolated cycles]
Let $G=(V,E, \head,\tail)$ be a multi--digraph $n \in \N$ and $e \in E$. If there  is a set $\mu \subseteq E$ of $n$ arcs, i.e., $|\mu|=n$ such that
$$ I_e\not = \emptyset \quad \text{w.r.t.} \quad G \setminus \mu$$
then we call the cycles $c \in \Oe(I_e)$ \emph{almost isolated cycles}. If $n=0$ then we obtain the notion of Definition \ref{def:iso}.
\end{definition}
As long as there are isolated cycles for small $n \in \N$ one can hope that the accuracy of TIGHT--CUT remains high.
We take this relaxed notion into account as follows.
If no isolated cycles were found then we generate $N$ graphs $H_i=G\setminus \mu_i$ by randomly deleting arcs $\mu_i \subseteq E$, $|\mu_i|=n$, $i=1,\dots,N$, $n,N \in \N$
and ask for the existence of  almost isolated cycles, i.e., search for  arcs $f$ in $H_i$ with $I_f\not = \emptyset$.
The arc appearing most in all the explored graphs $H_i$ is assumed to be a good choice for cutting it in the original graph.
If no such arc can be found then we use GOOD--GUESS for making a choice in any case.
The relaxation is formalized in Algorithm \ref{alg:tightR}.

\section{The Dualism of the  FVSP  \& FASP}
\label{Dual}
Though the dualism of the FVSP and the FASP is a known fact, its treatment is spread over the following publications \cite{ausiello,crescenzi1995compendium,Even:1998,FASP,kann1992}.
Here, we summarize and simplify the known results into one compact presentation allowing also to consider weighted versions. We recommend \cite{kann1992} for a modern introduction into approximation theory.

In our previous work \cite{FASP} 
we missed the crucial difference between the directed and undirected FVSP and therefore misleadingly used the dualism to 
claim the APX--completeness of the FASP. In addition to its new contributions, we want to take this section as a chance to correct our misunderstanding.

The following additional notions and definitions are required.

For a given vertex $v \in V$,
$\cev  N_E(v):=\li\{ e \in E \mi  \head(e)  = v\re\}$, $\vec N_E(v):=\li\{ e \in E \mi  \tail(e)  = v\re\}$, $N_E(v) = \cev{N}_E(v) \cup \vec N_E(v)$
shall denote the set of all incoming or outgoing arcs of $v$, and their unions.
The \emph{indegree} (respectively \emph{outdegree}) of a vertex $v$ is given by $ \cev \deg(v) = |\cev N_E(v)|$, $ \vec \deg(v) = |\vec N_E(v)|$ and the degree of a vertex is
$\deg(v)=|N_E(v)|$. The maximum degree of a graph is denoted by $\Delta(G):=\max_{v \in V} \deg(v)$.

\begin{definition}[directed line graph] The \emph{directed line graph} $\LL(G)=(V_L,E_L,\head_L,\tail_L)$ of a multi--digraph $G$ is a digraph where each vertex
represents one of the arcs of $G$, i.e., $V_L:= E$. Two vertices are connected by an arc
if and only if the corresponding arcs are consecutive, i.e., $  E_L :=  \li\{(e,f) \in E\times E \mi e,f \,\, \text{are consecutive} \re\}$, with
\begin{align*}
  \head_L,\tail_L :   E_L \lo V_L\,, \quad   \head_L\big((e,f)\big) := f\,, \,\, \tail_L\big((e,f)\big) := e\,.
\end{align*}
If there is an arc weight $\omega : E \lo \R^+$ on $G$ then we consider the induced vertex weight $\psi_L : V_L\lo \R^+$ given by $\psi(h) = \omega(h)$, for all $h \in V_L=E$.

\end{definition}

\begin{remark}
 Note that the line graph $\LL(G)$ has no multiple arcs and can be constructed in $\Oc(\Delta(G)|E|)$.
\end{remark}

The dual concept is to derive the \emph{natural hyper--graph} $\Hc(G)$  of a multi--digraph $G$ such that $G$ becomes the line graph of $\Hc(G)$, i.e., $\LL(\Hc(G))=G$. More precisely:

\begin{definition}[natural hyper--graph] Let $G=(V,E,\head,\tail)$ be a multi--digraph. We set
$V_H = E$ and introduce hyper--arcs
$E_H = \{e_v \mi v \in V\}$ with $\head_H(e_v)= \vec N_E(v)$, $\tail_H(e_v)= \cev N_E(v)$. The \emph{natural hyper--graph} is then given by $\Hc(G)=(V_H,E_H,\head_H, \tail_H)$.
See Figure~\ref{fig:hyp}  for an example. If there is a vertex weight $\psi : V \lo \R^+$ on $G$ then we consider the hyper--arc weight $\omega_H : E_H \lo \R^+$ given by
$\omega_H(e_v) = \psi(v)$ for all $e_v \in E_H$.
\end{definition}

\begin{remark}
 Note that for any multi-digraph $G$ the natural hyper--graph $\Hc(G)$ contains no multiple hyper--arcs and can be constructed in $\Oc(\Delta(G)|V|)$. Further, $\head(e_v)=\emptyset$, $\tail(e_v)=\emptyset$ is allowed.
 While  a loop $e \in E$ with $\Vc(e)=v$ in $G$  is represented in $\Hc(G)$ by the property $e \in \vec N_E(v),\cev N_E(v)$.
\end{remark}

\begin{definition}[dual digraph] Let $G=(V,E,\head,\tail)$ be a multi--digraph and $\Hc(G)=(V_H,E_H,\head_H,\tail_H)$ its natural hyper--graph.
For every hyper--arc $e_v \in E_H$ we consider the bipartite graph  $G_{e_v}=(V_{e_v},E_{e_v},\head_{e_v},\tail_{e_v})$ given by
$$V_{e_v} = \head_H(e_v) \cup \tail_H(e_v)\cup \{u,w\}\,, \quad  E_{e_v} =\{f_{v}\}\cup \big(\{u\}\times \tail_H(e_v)\big) \cup \big(\{v\}\times \head_H(e_v)\big)\,.$$
Further, we set $\head_{e_v}(f_{v}):=w\,, \tail_E(f_{e_v}):=u$ and
$$ \tail_{e_v}\big((x,y)\big) =x \,, \,\, \head_{e_v} \big((x,y)\big) =y \quad \forall \,\,(x,y) \in  \big(\{u\}\times \tail_H(e_v)\big) \cup \big(\{w\}\times \head_H(e_v)\big)\,.  $$
Finally, we consider the sets
$$V^*= \bigcup_{e_v \in E_H} V_{e_v} \,, \quad E^*= \bigcup_{e_v \in E_H} E_{e_v} \,,$$
define the maps $\head^*,\tail^*$ as the continuation of $\head_{e_v},\tail_{e_v}$ onto $E^*$ and denote the  dual multi--digraph of $G$ by $G^*=(V^*,E*,\head^*,\tail^*)$.
If there is a vertex weight $\psi: V \lo \R^+$ on $G$ then we consider the arc weight $\omega^* : E^* \lo \R^+\cup\{\infty\}$ given by
$$\omega^*(f)= \li \{\begin{array}{cl}
                \psi(v) &,  \text{if} \,\, f=f_{v} \,\, \text{for some}\,\,\, v \in V\\
                \infty &, \text{else}
               \end{array}\re. \,.
$$
Figure~\ref{fig:hyp} illustrates an example. The thin arcs of $G^*$ are weighted with $\infty$ and the non--filled vertices correspond to the artificially introduced vertices $\{u,w\}$.
\label{def:dualG}
\end{definition}
\begin{remark}\label{dual}
 Again we observe that the dual digraph possesses no multiple arcs  and can be constructed from $G$ in $\Oc(\Delta(G)|V|)$.
\end{remark}
Combining the definitions above we obtain maps
\begin{equation} \label{maps}
 \tau: V \lo E^* \,, \,\,\,\tau(v) = f_v \quad \text{and} \quad    \varrho : E \lo V_L\,, \quad \varrho(e) = e\,.
\end{equation}
\begin{figure}[t!]
 \centering
\includegraphics[width=15cm]{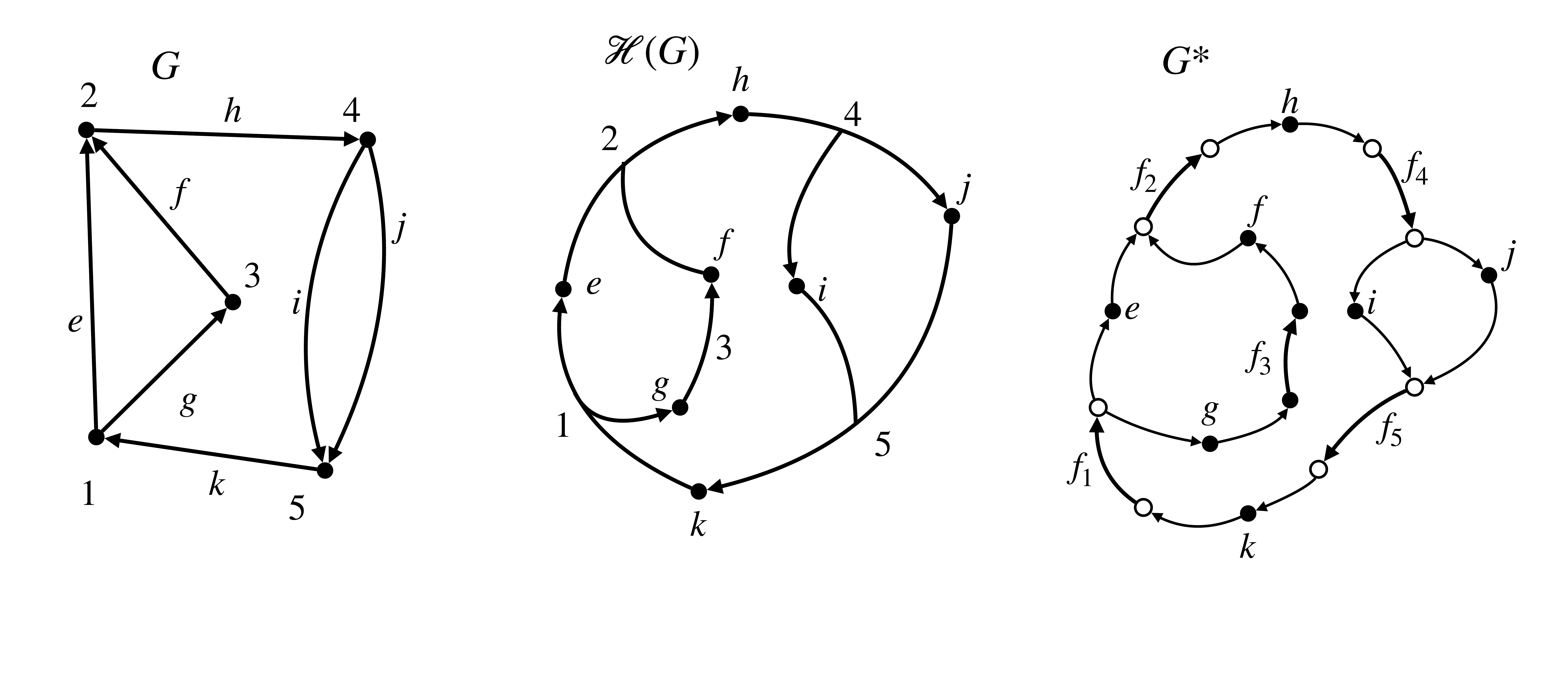}
\vspace{-1.5cm}
 \caption{The natural hyper--graph $\Hc(G)$ and the dual digraph $G^*$ of $G$.}
 \label{fig:hyp}
\end{figure}
Indeed $\tau$ and $\varrho$ allow to show that weighted FVSP and the the weighted FASP are approximation preservable reducible to each other.

\begin{proposition} \label{Lredu} Let $(G,\omega,\psi)$ be a weighted multi--digraph and $\ee \in \mathcal{F}_E(G,\omega)$ be a minimum feedback arc set and $\nu \in \mathcal{F}_V(G,\psi)$ be a minimum feedback vertex set of $G$.
\begin{enumerate}
 \item[i)] $\tau(\nu)$ is a minimum feedback arc set of  $(G^*,\omega^*)$ with $\Psi_{G^*,\omega^*}(\tau(\nu)) =\Psi_{G,\psi}(\nu)$.
 \item[ii)] $\varrho(\ee)$
is a minimum feedback vertex set of $(\LL(G),\psi_L)$ with $\Psi_{\LL(G),\psi_L}(\varrho(\ee)) = \Omega_{G,\omega}(\ee)$.
\end{enumerate}
\end{proposition}
\begin{proof}
We show that the cycles $\Oe(G)$, $\Oe(\LL(G)$ and $\Oe(G^*)$ are in $1:1$ correspondence. Indeed,
the vertex set $\Vc(d) = \{e_1,\dots,e_n\} =:c_d$ of any cycle $d\in \Oe(\LL(G))$ induces exactly
one cycle $c_d \in \Oe(G)$. Vice versa since $\LL(G)$ is a digraph without multiple arcs we observe that the arc set $\Ec(c)$ of any cycle $c \in \Oe(G)$ induces exactly one cycle $d_c\in \Oe(\LL(G)$ with $\Vc(d_c)=\Ec(c)$.
Analogously, we note that a cycle $d \in \Oe(G^*)$ is uniquely determined by knowing all its arcs $A(d)=\{f_v \in \Ec(d)\,, v \in V\}$ see Definition \ref{def:dualG}. Since $G^*$ possesses no multiple arcs this implies again that
the vertex set $\Vc(c)$ of any cycle $c \in \Oe(G)$ induces exactly one cycle $d_c \in \Oe(G^*)$ with $A(d_c)=\{f_v \in \Ec(d)\,, v \in \Vc(c)\}$. Vice versa, for any cycle $d \in \Oe(G^*)$ we have that
the vertex set $\Vc(d) \setminus \Vc(A(d))$ induces exactly one cycle $c_d \in \Oe(G)$ with $\Ec(c_d)=\Vc(d) \setminus \Vc(A(d))$. Hence $\Oe(\LL(G))\cong\Oe(G)\cong\Oe(G^*)$.
Consequently, for any $\ee \subseteq E$, $\nu \subseteq V$ there holds
\begin{align}
 \Ec(d) \cap \ee \not= \emptyset\,, \,\,\, \forall \, d\in \Oe(G) \quad  &\Longleftrightarrow \quad \Vc(c) \cap \varrho(\ee)\not= \emptyset\,, \,\,\,  \forall\, c\in \Oe(\LL(G)) \label{feedset1} \\
  \Vc(d) \cap \nu \not= \emptyset\,, \,\,\, \forall \, d\in \Oe(G) \quad &\Longleftrightarrow  \quad  \Ec(c) \cap \tau(\nu)\not= \emptyset\,, \,\,\,  \forall\, c\in \Oe(G^*)  \label{feedset2}
\end{align}
Since the identities $\Psi_{G^*,\omega^*}(\tau(\nu)) =\Psi_{G,\psi}(\nu)$, $\Psi_{\LL(G),\psi_L}(\varrho(\ee)) = \Omega_{G,\omega}(\ee)$ are a direct consequence of the definitions above due to \eqref{feedset1},\eqref{feedset2}
any $\ee \subseteq E$, $\nu \subseteq V$ is a minimum feedback arc/vertex set w.r.t. $G$ if and only if $\varrho(\ee) \subseteq E$, $\tau(\nu) \subseteq V$  is a minimum feedback vertex/arc set w.r.t. $\LL(G)$, $G^*$, respectively.
 \end{proof}

Consequently, we obtain the following well-known statement.

\begin{theorem} The (weighted) directed FVSP \& FASP are APX--hard.
 \label{hard}
\end{theorem}
\begin{proof}
Due to \cite{papa}
the \emph{minimum vertex cover problem} (MVCP) is known to be MAX SNP--complete. Since the class of APX--complete problem is given as the closure of MAX SNP under PTAS \cite{kann1992} the MVCP is APX--complete.
Already in \cite{Karp:1972} an approximation preserving $L$--reduction from the MVCP to the directed FVSP is constructed.

The graphs $\LL(G)$ and $G^*$ can be constructed from $G$ in polynomial time. Further, due to Lemma \ref{Lredu} the maps $\tau,\varrho$ from \eqref{maps} induce an approximation preserving
$L$--reduction from the FVSP to the FASP and vice versa, i.e.,
the FASP on $G$ is equivalent to the FVSP on $\LL(G)$ and the FVSP on $G$ is equivalent to the FASP on $G^*$. The second reduction implies the APX--hardness of the FASP. Since the weighted versions include the
case of constant weights the statement is proven.
\end{proof}

\begin{remark}
The reduction from MVCP to FVSP in \cite{Karp:1972} can be adapted also for the undirected FVSP.
Due to \cite{Dinur} the MVCP can not be approximated in polynomial time beneath a ratio of $r=10\sqrt{5} -21 \approx 1.3606$, unless P=NP. In light of this fact, and due to the circumstance
that the $L$--reductions from the MVCP to the FVSP and to the FASP are all approximation preserving the FVSP and the FASP
are also not polynomial time approximable beneath that ratio. If the \emph{unique games conjecture} is true then it is even impossible to approximate all three problems efficiently  beneath a ratio of $r=2$ \cite{khot}.
\end{remark}

\section{Previous Results}
\label{app:FASP}
We deliver the outstanding proofs of the statements in section \ref{sec:Iso}. As already mentioned these statements were already proven in our previous work \cite{FASP} and are given here in a simplified version.

\begin{theorem} \label{prop:appendix} Let $G=(V,E,\head,\tail)$ be a multi--digraph with arc weight $\omega : E \lo \R^+$.
\begin{enumerate}
\item[i)] There exist algorithms computing the subgraph $G_e$ in $\Oc(|E|^2+|V|)$ and $I_e$ in $\Oc(|E|+|V|)$.
 \item[ii)] If $\ee \in \mathcal{F}_E(G,\omega)$ is a minimum feedback arc set with  $e \in \ee$ then $\vec F(e) \subseteq \ee$.
 \item[iii)] If $I_e\not = \emptyset$  and  $\delta=\mathrm{mincut}(\head(e),\tail(e),I_e,\omega_{| I_e}) \in \Pc(E) $ be a minimum--$s$--$t$--cut with $s = \head(e)$, $t =\tail(e)$  w.r.t. $I_e, \omega_{|I_e}$ such that:
\begin{equation}\label{MCA}
 \Omega_{G,\omega}(\delta) \geq \Omega_{G,\omega}(\vec F(e)) \,.
\end{equation}
Then there is $\ee \in \mathcal{F}_E(G,\omega)$ with $\vec F(e) \subseteq \ee$.
\item[iv)] Checking whether $I_e \not =\emptyset$ and \eqref{MCA} holds can be done in $\Oc(|E||V|)$.
\end{enumerate}
\end{theorem}

  \begin{proof}  We show $i)$.
  Certainly, there has to be a directed path $p$ from $\head(e)$ to $\tail(j)$ and from $\head(j)$ to $\tail(e)$ for every arc $j \in \Ec(G_e)$.
 If $c \in O(G)\setminus \Oe(G)$ with $e \in \Ec(c)$ is a non--elementary cycle passing through $e$ then there is at least one arc $j \in \Ec(c)$ such that  $\head(j)$ or $\tail(j)$ are passed twice by $c$.
Hence, either there is no directed path $p$ from  $\head(e)$ to $\tail(j)$ in $G \setminus \vec N(\head(j))$ or there is no directed path $p$ from  $\head(j)$ to $\tail(e)$ in $G \setminus \cev N(\tail(j))$.
We denote with $J(c)$ all such arcs. If $f$ is an arc of an elementary cycle $c \in \Oe(G)$ then none of the cases occur, i.e., $f \not \in J(c)$, see Figure~\ref{fig:Ge}. Thus,
determining $J(c)$ can be done by running depth first search (DFS) at most $|E|$ times requiring $\Oc(|E|^2)$ operations. The strongly connected component $G'= \mathrm{SCC}_e(G)$ of $G \setminus J(c)$ that includes $\head(e)$ and $\tail(e)$ therefore coincides with $G_e$ and can be
be determined in $\Oc(|E|+|V|)$, \cite{SCC}. Now, we consider the set
$H(e) = \{\ f \in E_e \mi \Oe(f) \not = \emptyset \,\, \text{w.r.t.} \,\, G\setminus e \}$, which can be determined  by computing the SCCs of $G\setminus e$.
The SCC of $(G \setminus H(e))\cup\{e\}$ that includes $e$ yields $I_e$
finishing the proof.

We prove $ii)$. Let $\ee \in \mathcal{F}_E(G,\omega)$ and $e \in \ee$. Assume there is $f \in
\vec F(e)\setminus \ee$  then certainly $\ee \cap \cev F(e)= \cev F(e)$ otherwise there
would be a two--cycle that is not cut. Since $\ee \in \mathcal{F}_E(G,\omega)$ we have
$\Ec(c) \cap \ee \not = \emptyset$ for all $c \in \Oe(f)$
implying
$\Ec(c) \cap (\ee\setminus\{e\}) \not = \emptyset $ for all $c \in \Oe(e)$
contradicting the minimality of $\ee$. Thus, $\ee \cap \vec F(e) =\vec F(e)$.

To see $iii)$ we recall that by Remark \ref{rem:hierachy} there is no arc $f \in E$ with $I_f \supseteq I_e$ and $\omega(f) < \omega(e)$. On the other hand every other arc $f$ with $I_e=I_f$ satisfies $G_e =G_f$.
Due to  $G_e \supseteq I_e$ this implies that  if  $\Omega_{G,\omega}(\delta) \geq \omega(\vec F(e))$ then
$$\vec F(e) \in \mathcal{F}(I_e,\omega_{|Ie}) \quad \text{and by $i)$} \quad \Omega(G,\setminus \delta, \omega) \geq \Omega(G\setminus \vec F(e) ) \,.$$
Hence, we have proven $iii)$.

An exhaustive list of polynomial time algorithms with runtime complexity contained in $\Oc(|E|^3)$
computing  minimum--$s$--$t$--cuts is given in
\cite{MC,Goldberg2014}. Especially, for the unweighted case in \cite{Goldberg2014} an algorithm with $\Oc(|E||V|)$ or even faster is presented. A combination of \cite{KRT} and \cite{orlin} ensures that complexity
also for the weighted version.
Due to $i)$ this shows $iv)$.
\end{proof}

\bibliographystyle{plain}
\bibliography{Ref.bib}

\end{document}